\documentclass[11pt]{article}
\usepackage{fullpage}
\usepackage{amsmath,amssymb,amsfonts,amsthm,epsfig}

\usepackage{ifthen}

\newcommand{\bfp}{\mathbf{p}}
\newcommand{\bfh}{\mathbf{h}}
\newcommand{\bfx}{\mathbf{x}}
\newcommand{\bfr}{\mathbf{r}}

\newcommand{\prob}[2][]{\text{\bf Pr}\ifthenelse{\not\equal{}{#1}}{_{#1}}{}\!\left[#2\right]}
\newcommand{\expect}[2][]{\text{\bf E}\ifthenelse{\not\equal{}{#1}}{_{#1}}{}\!\left[#2\right]}

   \def\11{\mathbf{1}}

\def\bR{\mathbb{N}^*} \def\bN{\mathbb{Z}^*_{\ge 0}}
 \def\br{\mathbf{r}} \def\00{\mathbf{0}}
 \def\bN{\mathbb{N}} \def\calU{\mathcal{U}}

\newtheorem{theorem}{Theorem}
\newtheorem{example}{Example}
\newtheorem{corollary}[theorem]{Corollary}
\begin{document}

\title{An Axiomatic Approach to Block Rewards}

\author{Xi Chen\\ Columbia University\\ \texttt{xichen@cs.columbia.edu}
\and
Christos Papadimitriou\\ Columbia University\\ \texttt{christos@cs.columbia.edu} \and
Tim Roughgarden\\ Columbia University\\ 
\texttt{tr@cs.columbia.edu}
}

\maketitle
\begin{abstract}
  Proof-of-work blockchains reward each miner for one completed block
  by an amount that is, in expectation, proportional to the number of
  hashes the miner contributed to the mining of the block.  Is
  this {\em proportional allocation rule} optimal?  And in what sense?
  And what other rules are possible?  In particular, what are the
  desirable properties that any ``good'' allocation rule should
  satisfy?  To answer these questions, we embark on an axiomatic
  theory of incentives in proof-of-work blockchains at the time scale
  of a single block.  We consider desirable properties of allocation
  rules including: symmetry; budget balance (weak or strong);
  sybil-proofness; and various grades of collusion-proofness.  We show
  that Bitcoin's proportional allocation rule is the unique allocation
  rule satisfying a certain system of properties, but this does not
  hold for slightly weaker sets of properties, or when the miners are
  not risk-neutral.  We also point out that a rich class of allocation
  rules can be approximately implemented in a proof-of-work
  blockchain.
\end{abstract}
\section{Introduction}\label{s:intro}

The Bitcoin protocol was a remarkable feat: eleven years after its
sudden appearance \cite{Satoshi}, and without much adjustment and
debugging, it has been used by millions of people and has launched the
blockchain industry. Arguably, the most crucial and ingenious aspect
of its design lies in the {\em incentives} the protocol provides to its 
miners to participate and follow it faithfully.  
We believe it is of great importance and interest to understand and
scrutinize the incentives provided by blockchain protocols---and to
do so through the point of view and the methodology of Economic
Theory, the science of incentives.

Flaws in the incentives of a blockchain protocol can manifest
themselves at multiple timescales.  For longest-chain proof-of-work
blockchains like Bitcoin, the most well-studied incentive-based
attacks, such as selfish mining~\cite{EG14,SSZ16,KiayiasEC16} and transaction
sniping~\cite{CKWN16}, concern miners reasoning strategically over
multiple block creation epochs.  For example, in selfish mining, a
miner relinquishes revenue in the short term to achieve greater
revenue (in expectation) in the long run via a type of forking attack.

This paper studies incentive issues and potential deviations from
intended miner behavior at the most basic time scale, that of a single
block creation epoch.  We focus on the allocation of {\em block
  rewards}, which
drives the incentive structure in Bitcoin and many other similar
protocols.  The dominant paradigm in proof-of-work blockchains is to
fix a per-block reward, and for each block to allocate the entire
reward to whichever miner first solves a difficult cryptopuzzle.
Assuming that miners independently and randomly guess and check
possible solutions to the cryptopuzzle, the expected reward earned by
a miner is proportional to their share of the total contributed
computational power.

The proportional reward allocation scheme is a
simple, natural, and compelling idea, and in all evidence it works
quite well.  
{\em But is there any sense in which it is ``optimal''?}\ To answer,
one has to start by considering the whole spectrum of options; next
one must articulate appropriate desiderata; and finally, characterize
the full extent of possible solutions that satisfy these desiderata.
This is in line with the axiomatic methodology, which has been
traditionally employed in Economic Theory for the development of
utility theory~\cite{vNM}, of impossibility results in social
choice~\cite{Arrow}, as well as of cooperative game
theory~\cite{Shapley}, to name three salient examples.  The advantage
of the axiomatic approach is that through it one understands not only
the domain of possibilities, but also the costs of transgressing the
boundaries of this domain.  This is our focus in this paper.

\subsection{The Proportional Allocation Rule and Its Alternatives}

We formalize the question above through the concept of {\em
  allocation rules}, functions that map profiles of contributed
hashing power to profiles of expected block rewards.
The input to such a function is an $n$-tuple $\bfh =(h_1,h_2,\ldots,h_n)$
of positive integers, where $h_i$ is the hash rate contributed by
miner~$i$ in a given block creation epoch, and~$n$ is the number of
distinct miners (i.e., distinct public keys) that contribute a
non-zero hash rate.  When the epoch ends and a new block is
authorized, one unit of reward becomes available,
and the question
is, {\em how should it be allocated to the miners?}
The output of an allocation rule specifies an answer to this question,
in the form of an $n$-tuple $\bfp=(p_1,p_2,\ldots,p_n)$, where~$p_i$ is
the expected reward to miner~$i$.
As mentioned above,
the allocation rule corresponding to the Bitcoin protocol
is the {\em proportional rule}, with
\[
p_i = \frac{h_i}{\sum_{j=1}^n h_j}.
\]

Is the proportional rule ``optimal''?  Or, an even more basic
question: What alternatives to the proportional allocation rule are
possible?  At first blush, it might seem as though a proof-of-work
blockchain has no choice but to implement the proportional allocation 
rule, as presumably the probability distribution over which miner is
first to produce a cryptopuzzle solution will be proportional to the
contributed computational power.
However, as we discuss in
Section~\ref{ss:impl}, a rich class of allocation rules can in
principle be implemented (at least approximately) within a
proof-of-work blockchain.  The key idea is to wait for a large number
of solutions to a medium-difficulty cryptopuzzle, rather than a single
solution to a high-difficulty puzzle.  Each solution acts as a
single sample from the distribution proportional to miner hash rates,
and with a large enough number of samples all miners' hash rates can
be estimated with high accuracy.  
This, in turn, permits the
(approximate) implementation of a wide range of allocation rules, even
though the true miner hash rates are not a priori known to the
protocol.

\subsection{Properties of Allocation Rules}\label{ss:props}

The definition of an allocation rule is generic enough; the question
is, what kinds of properties should such allocation rules satisfy?  We
consider an array of possible properties, which a blockchain designer
may (or may not) require of an allocation rule.  The
first several properties are motivated by economic viability and fairness
rather than incentives per se.
\begin{itemize}  

\item {\em Non-negativity.}
Expected rewards (the $p_i$'s) should be nonnegative.  That is,
  the protocol cannot require payments from miners.

\item {\em Budget-balance.}
The protocol cannot be ``in the red,'' meaning
the sum of expected block rewards cannot exceed the unit of
  block reward available.  {\em Strong} budget-balance insists that
the entire unit of block reward is allocated, while {\em weak}
budget-balance allows the protocol to withhold some of the block
reward from miners.

\item {\em Symmetry.} The allocation rule should not depend on the
  names of the miners (i.e., their public keys), only on their
  contributed hash rates.

\end{itemize}

Finally, there are two further properties aimed at disincentivizing
certain behaviors by the miners that may be considered undesirable by
the blockchain designer:
\begin{itemize}

\item {\em Sybil-proofness.}
No miner can possibly benefit by creating many accounts and splitting
its mining power among them.

\item {\em Collusion-proofness.}
Two or more
  miners cannot benefit by pooling their mining resources and somehow
  splitting the proceeds.  This property has different variants
  depending on what types of payments between colluding miners are
  permitted; see Section~\ref{ss:axioms}.
\end{itemize}
The proportional allocation rule satisfies all five properties
(presumably by design), including the strong version of
budget-balance.

\subsection{Our Results}

We prove a number of characterization results that identify
which allocation rules satisfy which sets of desired properties.
We begin with
risk-neutral miners, who care only about expected rewards (and no
other details of the reward distribution).
Our first result is:
\begin{enumerate}

\item The proportional allocation rule is the unique allocation rule
  that, with risk-neutral miners, satisfies non-negativity, strong
  budget-balance, symmetry, sybil-proofness, and a weak form of
  collusion-proofness (Theorem~\ref{theo1}).

\end{enumerate}
That Bitcoin's block reward scheme is a singularly good idea hardly
comes as a surprise.  Nevertheless, we believe there is value in
formally articulating what makes it unique.  
Further, if one relaxes the requirements slightly, additional
allocation rules become possible.  For example:
\begin{enumerate}

\item [(2)] A family of allocation rules that we call the {\em generalized
    proportional allocation rules} constitute the only allocation
  rules that, with risk-neutral miners, satisfy non-negativity, weak budget-balance, symmetry,
  sybil-proofness, and a slightly stronger form of collusion-proofness
  (Theorem~\ref{theo2}).

\end{enumerate}

But is it reasonable to assume that miners are risk-neutral?  The
phenomenon of {\em mining pools} for Bitcoin and other
cryptocurrencies (see e.g.~\cite{L+15}) is a behavior which,
intuitively, aims to reduce the risk of each miner, and thus suggests
that miners in the real world are {\em risk-averse.}\footnote{There
  may also be other reasons to join a mining pool, for example to
  avoid the cost of maintaining a full node, but risk-aversion is
undoubtedly a first-order factor.  (Who is willing to wait an expected
twenty years for their first reward?)} 
We prove an {\em
  impossibility result} that makes this intuition precise: 
\begin{enumerate}

\item [(3)] If miners
are risk averse (equivalently, their utility is the expectation of a
strictly concave function of the reward), then there is {\em no}
non-zero allocation rule that is symmetric, (weakly) budget-balanced,
sybil-proof, and (weakly) collusion-proof (Theorem~\ref{theo3}).

\end{enumerate}
{This result suggests that mining pools as a form of collusion are unavoidable.}  In contrast:
\begin{enumerate}

\item [(4)] If miners are {\em
  risk-seeking}, then Bitcoin's 
proportional allocation rule satisfies (strong versions) of 
all of the desired properties
(Theorem~\ref{theo7}).  

\item [(5)]
A deterministic implementation
of the proportional rule---with the block reward split fractionally
between miners---satisfies all of the desired properties even with
risk-averse miners (Corollary~\ref{cor:poss}).  

\end{enumerate}
The deterministic implementation in~(5) can be viewed as a simulation
of the functionality of a mining pool inside the blockchain protocol
itself, analogous to the ``Revelation Principle'' from mechanism
design theory.\footnote{FruitChain~\cite{FruitChain} can be likewise
  interpreted as a lower-variance version of the proportional
  allocation rule, with ``fruits'' playing the role of
  medium-difficulty puzzle solutions.  Bobtail~\cite{Bobtail} gives
  still another implementation of the proportional allocation rule
  using multiple puzzle solutions; their primary motivation was to
  reduce the variance in time between consecutive blocks, rather than
  that of miners' rewards per se.}

\section{Model}\label{s:model}

\subsection{Allocation Rules}\label{ss:rules}

Let $\bR=\cup_{n\ge 1} \bN^n$ denote the set of all finite tuples of
positive integers.  For a positive integer $n$, we write $[n]$ to
denote $\{1,2,\ldots,n\}$.  
An \emph{allocation rule} is a function $\bfx$ over $\bR$ 
that maps each tuple $\bfh=(h_1,\ldots,h_n) \in \bN^n$ of length $n$ to 
an $n$-tuple $\bfp=(p_1,\ldots,p_n)$ of nonnegative real
numbers.
Here $h_i$ and $p_i$ are the hash rate and expected reward of
miner~$i$.\footnote{For example, $h_i$ could be in units of hashes per
  second.  For convenience, we assume this is an integer.  All of our
  results for integral hash rates immediately imply the same results
for arbitrary positive rational hash rates. 
(The proofs hold verbatim for hash rates that are integral multiples
of~$1/m$ for some $m \in \bN$, and hence apply to all finite rational
tuples of hash rates.)}
We also write $x_i(\bfh)$ for $p_i$ and $\bfx(\bfh)$ for
$\bfp$.
We refer to a tuple $\bfh \in \bR$ as a \emph{configuration} 
and $\bfx(\bfh)$ as the corresponding \emph{allocation}.

In this section and the next, we assume that miners care only about
their expected rewards (with more being better), and in particular are
risk-neutral.  Section~\ref{s:risk} addresses the case of
non-risk-neutral miners and the interesting new issues that they
raise.\footnote{With risk-neutral miners, there is no need to specify
details of the reward distribution beyond the expected reward for
each miner.  For example, $x_i(\bfh)$ might represent a deterministic 
reward of~$x_i(\bfh)$ to miner~$i$, or that miner~$i$ has
a~$x_i(\bfh)$ probability of winning the entire block reward (and
with the remaining probability receives no reward).
We'll be more specific about the semantics of an allocation rule in
due time, when
we consider non-risk-neutral miners in Section~\ref{s:risk}.}

\subsection{Axioms}\label{ss:axioms}

With the language and notation of allocation rules, we can translate the
properties in Section~\ref{ss:props} into formal
axioms.\footnote{Non-negativity is already baked into our definition
  of an allocation rule.}
See Section~\ref{ss:exs} for examples of allocation rules that satisfy
different subsets of these properties.\vspace{0.1cm}
\begin{enumerate}
\item [A1.]
\textbf{Symmetry}: An allocation rule $\bfx$ is \emph{symmetric} if  $\bfx(\pi(\bfh))=\pi(\bfx(\bfh))$
   for every configuration $\bfh\in \bR$ and 
 every permutation $\pi$.\vspace{0.1cm}
\end{enumerate}
That is, the expected reward of a miner 
does not depend on how the miners are ordered (or their public keys).
An example of an asymmetric rule is a {\em dictator} rule, which
always allocates the block reward to the miner with the (say)
lexicographically smallest public key.  We consider only symmetric
rules in this paper.\vspace{0.1cm}

\begin{enumerate}

\item [A2a.] \textbf{Strong budget-balance}: An allocation rule $\bfx$
  is \emph{strongly budget-balanced} if $\sum_{i} x_i(\bfh) = 1$
for every configuration $\bfh$.

\item [A2b.] \textbf{Weak budget-balance}: An allocation rule $\bfx$
  is \emph{weakly budget-balanced} if $\sum_{i} x_i(\bfh) \le 1$
for every configuration $\bfh$.

\item [A3.] \textbf{Sybil-proofness}:
An allocation rule $\bfx$ is {\em sybil-proof} if: For every
configuration $\bfh \in \bR$ and every configuration $\bfh'$ that can
be derived from~$\bfh$ by replacing a miner with hash rate $h_i$ by
a set~$S$ of miners with total hash rate at most $h_i$ (i.e., with
$\sum_{j \in S} h'_j \le h_i$), the total expected reward to miners
of~$S$ under $\bfh'$ is at most that of miner~$i$ in the original
configuration:
\[
\sum_{j \in S} x_j(\bfh') \le x_i(\bfh).\vspace{0.0cm}
\]

\end{enumerate}

We consider several natural definitions of collusion-proofness,
depending on the type of reward sharing allowed inside the
coalition and on whether a Pareto improvement for a coalition is
required to be strict.  We consider both arbitrary revenue-sharing
agreements and proportional sharing.  The latter corresponds to 
the reward schemes used in many Bitcoin mining pools (see
e.g.~\cite{L+15}).\vspace{0.1cm}
\begin{enumerate}

\item [A4a.] \textbf{Collusion-proofness (under arbitrary reward sharing)}:
An allocation rule $\bfx$ is {\em collusion-proof (under arbitrary
reward sharing)} if: For every
configuration $\bfh \in \bR$ and every configuration $\bfh'$ that can
be derived from~$\bfh$ by replacing a set~$T$ of miners with a new
miner~$i^*$ (representing the coalition) with hash rate at most the
total hash rate of miners in~$T$
(i.e., with $h'_{i^*} \le \sum_{j \in T} h_j$), 
the total expected
reward to miner~$i^*$ under $\bfh'$ is at most 
the total expected reward of miners of~$T$ in the original
configuration:
\[
x_{i^*}(\bfh') \le \sum_{j \in T} x_j(\bfh).
\]

\item[A4b.] \textbf{Strong collusion-proofness (under proportional
    sharing)}:
An allocation rule $\bfx$ is {\em strongly collusion-proof (under proportional
  sharing)} if: For every
configuration $\bfh \in \bR$ and every configuration $\bfh'$ that can
be derived from~$\bfh$ by replacing a set~$T$ of miners with a new
miner~$i^*$ with hash rate at most the total hash rate of miners in~$T$,
either: (i) no miner
of~$T$ has strictly higher expected reward in~$\bfh'$ (with
proportional sharing) than in~$\bfh$; or (ii)
some miner of~$T$ has strictly lower expected reward in~$\bfh'$
(with proportional sharing) than in~$\bfh$.  That is, if
\[
x_{i^*}(\bfh') \cdot \frac{h_i}{\sum_{j \in T} h_j} >  x_i(\bfh)
\]
for some miner~$i \in T$, then 
\[
x_{i^*}(\bfh') \cdot \frac{h_{\ell}}{\sum_{j \in T} h_j} <  x_{\ell}(\bfh)
\]
for some other miner~$\ell \in T$.

\item[A4c.] \textbf{Weak collusion-proofness (under proportional
    sharing)}: An allocation rule $\bfx$ is {\em weakly collusion-proof
    (under proportional sharing)} if: For every configuration
  $\bfh \in \bR$ and every configuration $\bfh'$ that can be derived
  from~$\bfh$ by replacing a set~$T$ of miners with a new miner~$i^*$
  with hash rate at most the total hash rate of miners in~$T$, some
  miner of~$T$ has expected reward under $\bfh'$ (with proportional sharing)
  at most that in the original configuration.  That is, for some
  miner~$i \in T$,
\[
x_{i^*}(\bfh') \cdot \frac{h_{i}}{\sum_{j \in T} h_j} \le  x_{i}(\bfh).
\]

\end{enumerate}
Every violation of~(4c) also constitutes a violation of~(4b), and
similarly for~(4b) and~(4a).  That is, (4a)--(4c) are ordered from
strongest (i.e., most difficult to satisfy) to weakest.
Impossibility and uniqueness results are most compelling for the
weakest variants of budget-balance and collusion-proofness;
possibility results are most impressive for the strongest variants.
The proportional allocation rule satisfies the strongest versions of
budget-balance and collusion-proofness (in addition to symmetry and sybil-proofness).

\subsection{Examples}\label{ss:exs}

This section gives several examples of allocation rules which satisfy
different subsets of the axioms in Section~\ref{ss:axioms}.

\begin{example}[Proportional allocation rule]\label{ex:1}
The proportional allocation rule is 
defined for each
configuration~$\bfh$ of length~$n$ by
\[
x_i(\bfh) = \frac{h_i}{\sum_{j \in [n]} h_j}.
\]
This allocation rule is symmetric, 
strongly budget-balanced,
sybil-proof, and collusion-proof under arbitrary reward sharing.
\end{example} 

\begin{example}[All-zero allocation rule]\label{ex:2}
The all-zero allocation rule is defined for each
configuration~$\bfh$ of length~$n$ by
\[
x_i(\bfh) = 0
\]
for~$i=1,2,\ldots,n$.
The all-zero allocation rule is symmetric, 
weakly budget-balanced, sybil-proof, 
and collusion-proof under arbitrary reward sharing.
\end{example} 

Our next example comprises scaled versions of the proportional
allocation rule, with the scaling constant dependent on the total hash
rate.
\begin{example}[Generalized proportional allocation rules]\label{ex:genprop}
For every nondecreasing function
$c:\mathbb{N}\rightarrow [0,1]$,
the corresponding generalized proportional allocation rule is
defined for each configuration~$\bfh$ of length~$n$ by
\[
x_i(\bfh) = c\left(\sum_{j \in [n]} h_j \right) \cdot \frac{h_i}{\sum_{j \in
  [n]} h_j}.
\]
Examples \ref{ex:1} and \ref{ex:2} are the generalized proportional
allocation rules corresponding to the functions $c(y)=1$ and $c(y)=0$,
respectively.

Generalized proportional allocation rules are symmetric,
weakly budget-balanced, sybil-proof, and
collusion-proof under arbitrary reward sharing.
To check sybil-proofness, let $\bfh$ be a configuration and $\bfh'$
derived from~$\bfh$ by replacing some miner~$i$ with hash rate $h_i$ by
a set~$S$ of miners with total hash rate at most $h_i$ (i.e., with
$\sum_{j \in S} h'_j \le h_i$).
Then, the total expected reward under $\bfh'$ of the miners in~$S$ is:
\begin{align*}
c \left(\sum_{j\in [n]}  h_j-h_i+\sum_{j\in S} h'_j\right)&\cdot 
\frac{\sum_{j\in S} h'_j}{\sum_{j\in [n]} h_j-h_i+\sum_{j\in S} h'_j}\\[0.4ex]
&\hspace{-3.5cm}=c \left(\sum_{j\in [n]} h_j-h_i+\sum_{j\in S} h'_j\right)  
\left(1-\frac{\sum_{j\in [n]}h_j-h_i}{\sum_{j\in [n]} h_j-h_i+\sum_{j\in S} h'_j}\right)\\[0.8ex]
&\hspace{-3.5cm}\le c \left(\sum_{j\in [n]} h_j \right)  
\left(1-\frac{\sum_{j\in [n]}h_j-h_i}{\sum_{j\in [n]} h_j}
  \right)\\[0.8ex]
         &\hspace{-3.5cm}
=c \left(\sum_{j\in [n]} h_j \right)\cdot \frac{h_i}{\sum_{j\in [n]} h_j},
\end{align*}
with the inequality following from the fact that
$\sum_{j\in S} h'_j \le h_i$.
Since the final expression is miner~$i$'s expected reward in the
original configuration, this verifies sybil-proofness.
Collusion-proofness can be checked using a similar argument.
\end{example} 

\begin{example}[Proportional-to-squares allocation rule]
The propor\-tional-to-squares allocation rule is 
defined for each
configuration~$\bfh$ of length~$n$ by
\[
x_i(\bfh) = \frac{h^2_i}{\sum_{j=1}^n h^2_j}.
\]
This allocation rule is symmetric, strongly budget-balanced, and
sybil-proof.  It is not even weakly collusion-proof under proportional
sharing.
\end{example} 

\begin{example}[Proportional-to-square-roots allocation rule]
The proportional-to-square-roots allocation rule is 
defined for each configuration~$\bfh$ of length~$n$ by
\[
x_i(\bfh) = \frac{\sqrt{h_i}}{\sum_{j=1}^n \sqrt{h_j}}.
\]
This allocation rule is symmetric, strongly budget-balanced, and
collusion-proof under arbitrary reward sharing.
However, it is not sybil-proof.
\end{example}

\subsection{Implementing Non-Proportional Rules}\label{ss:impl}

The proportional allocation rule is realizable, in that there is a
proof-of-work blockchain protocol (namely, Bitcoin) that implements
it.  Are non-proportional allocation rules purely hypothetical?  
This section demonstrates that, at least in principle,
{\em   every} weakly budget-balanced and continuous
allocation rule can be approximately implemented within a
proof-of-work blockchain.  

Fix a power-of-2~$M$, a cryptographic hash
function~$f$, and a difficulty level~$b$.  
A {\em full solution} is a preimage~$z$ such that $f(z)$ has at least
$b$ trailing zeroes, and a {\em partial solution} is a~$z$ such that
$f(z)$ has at least $b - \log_2 M$ trailing zeros.
The parameter~$b$ is chosen so that random guessing by miners 
produces a full solution in a
prescribed amount of time (on average), such as 10 minutes.  
One expects partial solutions to be discovered at~$M$ times the rate of
full solutions.

Fix a weakly budget-balanced allocation rule~$\bfx$.
A high-level description of one possible corresponding protocol is then:\vspace{0.1cm}
\begin{enumerate}

\item Miners attempt to find partial and full solutions of the form
  $\langle pkey | \sigma | nonce \rangle$, where $pkey$ is the miner's
  public key, $\sigma$ is derived from the current blockchain state
  (e.g., the hash of the block of transactions at the end of the
  longest chain), and $nonce$ is a number of free bits specified by
  the miner.

\item A miner who discovers a partial (non-full) solution can add
  it to the blockchain.

\item A miner who discovers a full solution can authorize a new block
  of transactions and add it to the blockchain (along with their full
  solution).

\item When a new full solution and corresponding block are published,
  block rewards are distributed according to the number of partial
  solutions contributed by each miner since the preceding full
  solution.  
Precisely, let $[n]$ denote the miners contributing at least one
partial solution, $g_i$ the number contributed by miner~$i$, and~$M' =
\sum_{i \in [n]} g_i$ the total number of partial solutions reported in
  this epoch.  Define an estimate~$\smash{\hat{h}_i}$ of miner~$i$'s hash rate
  by
\[
\hat{h}_i = g_i \cdot \rho \cdot \frac{M}{M'},
\] 
where~$\rho$ denotes the hash rate that would produce (on average) one
partial solution in the prescribed amount of time.
Define miner rewards by
  $\smash{\bfx(\hat{h}_1,\hat{h}_2, \ldots,\hat{h}_n)}$.\vspace{0.1cm}

\end{enumerate}
Several comments are in order.  First, this protocol effectively simulates the
typical functionality of a mining pool {\em inside the protocol
  itself}.  This is reminiscent of the {\em Revelation Principle} from
mechanism design (see e.g.~\cite{AGTbook}), which is a simulation
argument that shows how to eliminate non-truthful reporting of
preferences by moving the deviations inside the mechanism itself.

Second, if we take~$M=1$, then partial and full solutions coincide,
$M=M'=1$, and the protocol essentially recovers Bitcoin's
implementation of the proportional allocation rule.

Third, the larger we take~$M$, the more accurate the estimated hash
rates $\smash{\hat{\bfh}}$.  Thus any continuous allocation rule can be
approximated as closely as desired by taking~$M$ sufficiently large.

Fourth, we ignore incentive issues involving miners' delaying the
publication of full or partial solutions; Schrivjers et
al.~\cite{fc16} discuss (in the context of mining pools) methods for
mitigating such incentive issues.

Finally, we have deliberately avoided committing to the details of the
implementation, such as the best way to record the partial solutions
on-chain. 
Our point is simply that there appears to be no fundamental
barrier to implementing a wide range of non-proportional allocation
rules in a proof-of-work blockchain.

\section{Characterizations with Risk-Neutral Miners}\label{s:unique}

\subsection{A Uniqueness Result for the Proportional Allocation
  Rule}\label{ss:unique}

We assume throughout this section that miners are risk-neutral and are
concerned only with their expected rewards.
Our first main result is a uniqueness result for the proportional
allocation rule: It is the only rule that is symmetric, strongly
budget-balanced, sybil-proof, and collusion-proof (even weakly
collusion-proof with proportional sharing).

\begin{theorem}[Characterization of Rules with A1, A2a, A3, A4c]\label{theo1}
  The proportional allocation rule is the unique allocation rule that
  is symmetric, strongly budget-balanced, sybil-proof, and weakly
  collusion-proof (with proportional sharing).
\end{theorem} 

\begin{proof} 
Let $\bfx$ be an allocation rule that is symmetric, 
strongly budget-balanced,
sybil-proof, and
weakly collusion-proof (with proportional sharing).
We prove by induction that for every configuration $\bfh \in \bN^*$,
  $\bfx$ satisfies $x_i(\bfh) = h_i/\sum_{j} h_j$. 
The induction is on the number $t$ of entries in $\bfh$ that are
larger than $1$.  The base case of $t=0$ is trivial since $\bfh$ is then
an all-$1$ tuple (of length~$n$, say) and symmetry and strong
budget-balance imply that $x_i(\bfh) = 1/n$ for every~$i$.
  
For the inductive step, assume that the statement holds for all
tuples with less than $t$ entries larger than $1$. Now consider a
configuration 
$\bfh \in \mathbb{N}^n$ for some $n\ge t$ that sums to $m$ and has $t$
entries larger than $1$.  
Because $\bfx$ is symmetric, we can assume
without loss of generality that $h_1,h_2,\ldots,h_t>1$ while
$h_{t+1},h_{t+2},\ldots,h_{n}=1$. 
Let  $\bfp=(p_1,\ldots,p_n)=\bfx(\bfh)$. 
We claim that
$p_i= {h_i}/{m}$ for each $i\in [t]$.
This claim, together with the assumptions that $\bfx$ is symmetric and
strongly budget-balanced, implies that $x_i(\bfh)=1/m$ for all $i>t$ and
hence $\bfx$ indeed agrees with the proportional rule on $\bfh$.

Suppose $p_i> {h_i}/{m}$ for some $i \in [t]$.
Then, consider the configuration $\bfh'$ obtained by splitting 
the miner $i$ into $h_i$ sybils, each with hash rate $1$.
By the inductive hypothesis and the fact that $\bfh'$ has $t-1$
entries that are larger than $1$, 
$\bfx$ is proportional on $\bfh'$ and so each of these miners with
hash rate $1$ receives $1/m$ which is strictly less than
${p_i}/{h_i}$.  This means that, in $\bfh'$,
these $h_i$ players with hash rate $1$ would all be better off by
colluding and sharing their results proportionally (i.e., uniformly).
This contradicts the assumption that $\bfx$ is weakly collusion-proof
(with proportional sharing).  We conclude that $p_i \le
{h_i}/{m}$ for all $i \in [t]$.

On the other hand, suppose $p_i<{h_i}/{m}$ for some $i \in [t]$,
and consider the configuration $\bfh'$ obtained by splitting miner~$i$
into~$h_i$ sybils with hash rate 1 each.  By the inductive hypothesis,
each sybil receives expected reward ${1}/{m}$ under~$\bfx$
in~$\bfh'$.  The total expected reward earned by the sybils therefore
exceeds that of miner~$i$ in~$\bfh$.  This contradicts the assumption
that $\bfx$ is sybil-proof, so we can conclude that
$p_i \ge {h_i}/{m}$ (and hence~$p_i = {h_i}/{m}$) for all
$i \in [t]$.  This completes the proof of the claim, the inductive
step, and the theorem.
\end{proof}

\subsection{Weak Budget-Balance and Generalized Proportional Rules}\label{ss:generalized}

Example~\ref{ex:genprop} shows that relaxing the strong budget-balance
requirement to weak budget-balance enlarges the design space.  One
might suspect that, analogous to Theorem~\ref{theo1}, generalized
proportional allocation rules are the
only ones that satisfy symmetry, weak budget-balance, sybil-proofness,
and weak collusion-proofness with proportional sharing.  The next example
shows that this is not the case.

\begin{example}\label{ex:half}
Consider the following allocation rule $\bfx$.
For a configuration~$\bfh$ in which no miner has more than half the
overall hash rate (i.e., $\smash{h_i \le \tfrac{1}{2} \sum_{j \in [n]} h_j}$
for every $i$), $\bfx$ agrees with the proportional allocation rule.
For a configuration in which one miner~$i$ has more than half of the
overall hash rate, miner~$i$ receives its fair share under the proportional rule (i.e., $x_i(\bfh)
= h_i/\sum_{j \in [n]} h_j$), while other miners receive~0 ($x_j(\bfh)
= 0$ for $j \neq i$).

The rule $\bfx$ is symmetric, weakly budget-balanced, 
sybil-proof, and weakly collusion-proof under proportional sharing.
It is not strongly collusion-proof under proportional sharing, however.
\end{example}

Our second main result shows that if weak collusion-proofness is
strengthened to strong collusion-proofness (with proportional
sharing), then generalized proportional allocation
rules are indeed the only ones that satisfy the axioms~A1, A2b, A3,
and A4b.

\begin{theorem}[Characterization of Rules with A1, A2b, A3, A4b]\label{theo2}
Generalized proportional allocation rules  are the only
allocation rules that are symmetric,
weakly budget-balanced, sybil-proof, and
strongly collusion-proof under proportional sharing.
\end{theorem}

\begin{proof} 
  Let $\bfx$ be an allocation rule that is symmetric, weakly
  budget-balanced, sybil-proof, and strongly collusion-proof under
  proportional sharing.  We claim that,
for each positive integer~$m$,
  there is a nonnegative real number
  $c(m) \le 1$ such that
\[
x_i(\bfh) = c(m) \cdot \frac{h_i}{m}
\]
for every configuration $\bfh$ with total hash rate ($\sum_j h_j$)
equal to~$m$.

To prove the claim,
fix a positive integer $m$.
First, let $1^m\in \mathbb{N}^m$ denote the all-$1$ tuple of length $m$ and
define $$c(m) = \sum_{i \in [m]} x_i(1^m).$$
Since $\bfh$ is symmetric and weakly budget-balanced,
$c(m)\le 1$ and $x_i(1^m) = {c(m)}/{m}$ for every~$i \in [m]$.
Second,
let $\bfh^*$ denote the configuration with a single miner with hash
rate $m$.  Since $\bfx$ is sybil-proof and strongly collusion-proof
(under proportional sharing), the expected reward assigned by $\bfx$
to the sole miner in~$\bfh^*$ is~$c(m)$.

We follow the approach
in the proof of Theorem \ref{theo1} and prove by induction that, for every
configuration $\bfh\in \bN^*$ with $\sum_i h_i = m$,
\[
x_i(\bfh) = c(m) \cdot \frac{h_i}{m}
\]
for every miner~$i$.
The induction is again on the number $t$ of entries in $\bfh$ that are 
larger than $1$.  The base case of $t=0$ is trivial by the choice of
$c(m)$.
  
For the inductive step,
we consider a configuration $\bfh \in \mathbb{N}^n$ for some $n\ge t$
with $\sum_{i \in [n]} h_i = m$ and exactly $t$ entries larger than $1$.
Without loss of generality (since $\bfx$ is symmetric), assume that $h_1,\ldots,h_t>1$ and $h_{t+1},\ldots,h_{n}=1$. 
Let $\bfp=(p_1,\ldots,p_n)=\bfx(\bfh)$.  With respect to the first~$t$
miners, the same argument as in the proof of Theorem~\ref{theo1}
shows that
$p_i= c(m) \cdot ({h_i}/{m})$ for every $i\in [t]$.
(This part of the argument requires only weak collusion-proofness
under proportional rewards.)
Example~\ref{ex:half} shows that the rest of the proof (for miners~$t+1,t+2,\ldots,n$) must differ from
that of Theorem~\ref{theo1} and make use of strong
collusion-proofness.  The issue is that
because $\bfx$ is only assumed to be weakly budget-balanced,
it does not follow directly from $p_i= c(m)\cdot ({h_i}/{m})$ for
  all $i \in [t]$ that $p_i={c(m)}/{m}$ for all $i>t$.
  
So assume for contradiction that $p_i\ne {c(m)}/{m}$ for some
$i>t$ (and hence, by symmetry, for all such~$i$).
If $p_i>{c(m)}/{m}$, then $\bfx$ fails sybil-proofness:
the sole miner in configuration~$\bfh^*$ can increase its expected
reward by splitting into $n$ sybils with hash rates as in $\bfh$.
On the other hand, if $p_i< {c(m)}/{m}$, then
$\bfx$ fails strong collusion-proofness under proportional sharing:
if the $n$ miners in $\bfh$ form the grand coalition, then with
proportional sharing, every miner $i\in [t]$ receives the same
expected reward $p_i=c(m)\cdot ({h_i}/{m})$ as in $\bfh$
while the expected reward of every miner $i>t$ 
strictly increases (from~$p_i$ to ${c(m)}/{m}$).
We conclude that $p_i = {c(m)}/{m}$ for every $i > t$ and
hence $p_i = c(m) \cdot ({h_i}/{m})$ for every miner~$i$.  This
concludes the proof of the claim.

All that remains is to show that~$c(m)$ must be a nondecreasing function
of~$m$.  This follows from sybil-proofness: 
if $c(m+1)<c(m)$ for some $m$, then in the single-miner configuration
with hash rate $m+1$, the miner would have an incentive to replace
itself with a miner with hash rate~$m$.
\end{proof}

\section{Beyond Risk Neutrality:  Possibility and Impossibility Results}\label{s:risk}

We have been assuming so far that miners are {\em risk-neutral} and
care only about the {\em expectation} of their reward, as opposed to
other distributional characteristics (like variance).  Going beyond
this assumption reveals the interesting ways in which risk affects
incentives in blockchain protocols.

\subsection{Von Neumann-Morgenstern Utilities}\label{ss:vnm}

The earliest, and most principled, treatment of risk in economics is
through von Neumann and Morgenstern's {\em utility theory,}
articulated more than seven decades ago \cite{vNM}.  One starts from
each agent having a very general set of arbitrary preferences between
{\em lotteries} (finite-support probabilistic distributions over
different amounts of money), where the preferences of agents are
assumed to satisfy four very natural and plausible axioms:
completeness, transitivity, continuity, and independence. The
remarkable result proved is that any such system of preferences of an
agent is tantamount to the agent possessing {\em a utility function
  $U$,} a function mapping amounts of money to the reals, such that
the agent prefers lottery $L$ to lottery $M$ if and only if
$\expect[L]{U(p)} \geq \expect[M]{U(p)}$, where the expectation is
over the random variable~$p$ (which is in units of money).  That is,
any agent with preferences satisfying the properties above can be
modeled as an {\em expected utility-maximizer}.

For an agent with utility function~$U$, we can interpret $U(p)$ 
as the amount of {\em utility} the agent receives from a reward of $p$.
For ease of presentation, we assume throughout that $U(0)=0$ and
that $U$ is twice-differentiable and strictly increasing.
The risk sensitivity of an agent can then be gauged by the second
derivative of $U$.
Risk-neutrality corresponds to a linear utility function, with~0
second derivative; in this case, $\expect{U(p)} = U(\expect{p})$ for
every distribution over rewards~$p$.
A {\em risk-averse} agent is an expected utility maximizer with a
strictly concave utility function---a function with an everywhere
negative second derivative.  For a risk-averse agent, $\expect{U(p)}
\le U(\expect{p})$ for every distribution over~$p$ (by Jensen's
inequality).
For example, a risk-averse agent may not risk flipping a fair coin if
the two outcomes are either losing $30\%$ of their fortune, or
doubling it.  (Whereas a risk neutral --- or {\em risk-seeking}, see
Section~\ref{ss:seek} --- agent would be happy to flip the coin.)

\subsection{Collusion-Proofness Revisited}\label{ss:axioms2}

In this section and the next, we consider (randomized) allocation
rules that allocate the entire reward to (at most) one miner.  We
interpret an output $x_i(\bfh)$ of an allocation rule
as the probability that miner~$i$
receives the entire reward in the configuration~$\bfh$.

Of the four types of axioms in Section~\ref{ss:axioms}, symmetry and
budget-balance are obviously independent of any model of miner
utility.  The sybil-proofness axiom requires only cosmetic changes:\vspace{0.15cm}
\begin{enumerate}

\item [A3.] \textbf{Sybil-proofness (with general utility functions)}:
  An allocation rule $\bfx$ is {\em sybil-proof} if:
For every allowable miner utility function~$U$,
every
configuration $\bfh \in \bR$, and every configuration $\bfh'$ that can
be derived from~$\bfh$ by replacing a miner with hash rate $h_i$ by
a set~$S$ of miners with total hash rate at most $h_i$ (i.e., with
$\sum_{j \in S} h'_j \le h_i$), the total expected utility to~$i$ in
$\bfh'$ is at most that in the original configuration:
\[
\sum_{j \in S} x_j(\bfh')\cdot U(1) \le  x_i(\bfh)\cdot U(1).
\]
\end{enumerate}
Assuming that every allowable utility function~$U$ is strictly
increasing (and thus, $U(1)>U(0)=0$), this axiom is equivalent to the definition of
sybil-proofness for risk-neutral miners given in Section~\ref{ss:axioms}.

We need to adjust the definition of
collusion-proofness, since now the incentives of miners
are affected by their utility functions. 
Fix a class $\calU$
of allowable utility functions --- for example, all linear functions, or all
strictly concave functions.
We next
define analogs of axioms~(4a) and~(4c) from Section~\ref{ss:axioms}
(we skip~(4b) because we do not need it in our statements).
A {\em reward-sharing scheme} with miner set~$S$ specifies how the
miners of~$S$ would split a block reward internally; formally, it is
a collection of nonnegative random variables $\{ \bfr_i \}_{i \in S}$
that satisfy $\sum_{i \in S} \bfr_i \le 1$ with probability~1.\vspace{0.15cm}

\begin{enumerate}

\item [A4a.] \textbf{Collusion-proofness against $\calU$ (under arbitrary reward
    sharing)}: An allocation rule $\bfx$ is {\em collusion-proof
    against $\calU$ (under arbitrary reward sharing)} if:
  For every configuration
  $\bfh \in \bR$, every choice of utility function $U_i \in \calU$ for
  each participating miner~$i$, every configuration $\bfh'$ that can
  be derived from~$\bfh$ by replacing a set~$T$ of miners with a new
  miner~$i^*$ with hash rate at most the total hash rate of miners
  in~$T$, and every reward sharing scheme~$\{ \bfr_i \}_{i \in T}$,
  either: (i) no miner of~$T$ has higher expected utility 
  in~$\bfh'$ (with the given reward sharing scheme) than in~$\bfh$; or (ii) some
  miner of~$T$ has strictly lower expected utility 
  in~$\bfh'$ (with the given reward sharing scheme) than in~$\bfh$.  That is, if
\[
x_{i^*}(\bfh') \cdot \expect{U_i(\bfr_i)}  >  U_i(1) \cdot x_i(\bfh)
\]
for some miner~$i \in T$, then 
\[
x_{i^*}(\bfh') \cdot \expect{U_{\ell}(\bfr_{\ell})}  <  U_{\ell}(1) \cdot x_{\ell}(\bfh)
\]
for some other miner~$\ell \in T$.

\item[A4c.] \textbf{Weak collusion-proofness against~$\calU$ (under proportional
    sharing)}: An allocation rule $\bfx$ is {\em weakly collusion-proof
against $\calU$    (under proportional sharing)} if:
  For every configuration
  $\bfh \in \bR$, every choice of utility function $U_i \in \calU$ for
  each participating miner~$i$, every configuration $\bfh'$ that can
  be derived from~$\bfh$
 by replacing a set~$T$ of miners with a new miner~$i^*$
  with hash rate at most the total hash rate of miners in~$T$, some
  miner of~$T$ has expected utility under $\bfh'$ (with proportional sharing)
  at most that in the original configuration.  That is, for some
  miner~$i \in T$,
\[
x_{i^*}(\bfh') \cdot U_i\left(\tfrac{h_i}{\sum_{j \in T} h_j}
  \right)  \le  U_i(1) \cdot x_i(\bfh).
\]

\end{enumerate}

\subsection{Risk Aversion and Impossibility}\label{ss:imposs}

With risk-neutral miners, the proportional allocation rule implemented
by the Bitcoin protocol satisfies all of the axioms in
Section~\ref{ss:axioms}, including collusion-proofness (even with
arbitrary reward sharing).  In reality, however, Bitcoin is {\em not}
collusion-proof, in that most miners join a mining pool that
effectively acts like a large single miner.\footnote{See
  \texttt{https://www.blockchain.com/en/pools} for the biggest Bitcoin
  mining pools.  As of this writing, the largest pool (BTC.com)
  controls roughly 20\% of the total hash rate.}  There is anecdotal
evidence that miners generally join mining pools to lower the variance
of the rewards received, analogous to an insurance policy, and we view
the ubiquity of mining pools as strong evidence that most miners are
risk-averse.  How does risk aversion affect the set of possible
allocation rules satisfying our standard axioms?

The following {\em impossibility result} shows that the incentive to
form mining pools
is not an artifact of the specific allocation rule implemented in
Bitcoin; rather, it is a fundamental difficulty  with risk-averse
miners.

\begin{theorem}[Impossibility with Risk-Averse Miners]\label{theo3}
Assume that all miners are expected utility-maximizers with the same
strictly concave utility function~$U$ satisfying $U(1)>U(0)=0$.
There is no non-zero allocation rule that is symmetric,
weakly budget-balanced,
sybil-proof, and 
weakly
collusion-proof against $\calU=\{U\}$ (with proportional sharing). 
\end{theorem}

Note that this impossibility result holds even with our weakest notions
of budget-balance~(A2b) and collusion-proofness~(A4c).
The assumption that all miners share the same utility function $U$ only makes
  the impossibility result more compelling.

\begin{proof}[Proof of Theorem \ref{theo3}]
  Let~$\bfx$ be a non-zero and symmetric allocation rule.  Let~$\bfh$
  be a configuration such that $x_i(\bfh) > 0$ for some~$i$, and
  assume for simplicity that $h_i$ is even.  Obtain~$\bfh'$ from
  $\bfh$ by replacing miner~$i$ with two miners $i_1$ and $i_2$ with
  hash rate $h_i/2$ each.  By symmetry,
  $x_{i_1}(\bfh') = x_{i_2}(\bfh')$; let~$q$ denote this common
  probability.  If $\bfx$ is sybil-proof, then $2q \le x_i(\bfh)$;
  suppose this is indeed the case.  Starting now from $\bfh'$, if the two
  miners~$i_1$ and~$i_2$ join forces and combine hash rates to produce
  the configuration $\bfh$ (sharing rewards proportionally), then
  their expected utilities change from
\[
U(1) \cdot q
\]
to
\[
U\left( \tfrac{1}{2}\right) \cdot x_i(\bfh) \ge 2 \cdot U\left(
  \tfrac{1}{2}\right) \cdot q.
\]
Our assumption on~$U$ implies that $2U(1/2) > U(1)$ and hence
both miners are strictly better off in the coalition.  We conclude
that $\bfx$ is not weakly collusion-proof under proportional sharing.
\end{proof}

\subsection{Possibility with Deterministic Rewards}\label{ss:poss}

The impossibility result in Theorem~\ref{theo3} holds for randomized
allocation rules $\bfx$ that always allocate the entire block reward
to a single miner according to the probability distribution specified
by $\bfx(\bfh)$.  At the other extreme are {\em deterministic}
allocation rules, for which each $x_i(\bfh)$ represents a (fractional)
block reward deterministically given to miner~$i$.
With risk-neutral miners, there is no difference between the two types
of rules, and both the randomized and deterministic implementations of
the proportional allocation rule satisfy all of our axioms.  

With a deterministic allocation rule, the miner utility functions no
longer matter---more (deterministic) reward is always better than
less.  (Remember that utility functions are strictly increasing.)
In this case, miners effectively act as if they were risk-neutral,
and so all results for risk-neutral miners carry over to deterministic
allocation rules and miners with arbitrary utility functions.  For
example, the following corollary of Theorem~\ref{theo1} is immediate.
\begin{corollary}[Possibility with Deterministic Rewards]\label{cor:poss}
For every family $\calU$ of utility functions,
the deterministic implementation of the proportional allocation rule
is the unique deterministic allocation rule that
is symmetric, strongly budget-balanced, sybil-proof, and weakly
collusion-proof against $\calU$ (with proportional sharing).
\end{corollary}
Bitcoin's implementation of the proportional allocation rule is
inherently randomized.  But a deterministic version of the rule can be
approximated as closely as desired using a protocol of the type
described in Section~\ref{ss:impl} (with a sufficiently large value of
the parameter~$M$).

In the proposed implementation in Section~\ref{ss:impl}, the choice of
the miner who can authorize a new block of transactions and add it to
the blockchain---the first miner to find a full solution---is still
effectively chosen randomly (see step~(3)).
This highlights the fact that there are really two distinct problems
to solve in each block creation epoch:
\begin{enumerate}

\item Electing a leader to authorize the next block of transactions.

\item Distributing block rewards. 

\end{enumerate}
In Bitcoin, the solutions to these problems are tightly coupled, in
that the choices of the leader and of the recipient of the block reward
are always the same.  The solutions are decoupled in
the deterministic implementation of the
proportional allocation, with the leader elected randomly as in
Bitcoin but with rewards distributed deterministically.

\subsection{The Case of Risk-Seeking Miners}\label{ss:seek}

For completeness, this section studies {\em risk-seeking} miners,
meaning expected utility maximizers with strictly convex utility
functions.  In a surprising contrast with the impossibility result for
risk-averse miners (Theorem~\ref{theo3}), 
here the (randomized implementation of the) proportional
allocation rule satisfies all of the same axioms as in the case of
risk-neutral miners.
  
\begin{theorem}[The Proportional Rule with Risk-Seeking Miners]\label{theo7}
Let $\calU$ be the set of all convex functions $U:[0,1]\rightarrow
\mathbb{R}_{\ge 0}$ with $U(1)>U(0)=0$.
The proportional allocation rule is symmetric, strongly
budget-balanced, sybil-proof, and collusion-proof against  $\calU$
under arbitrary reward sharing.
\end{theorem}

\begin{proof}
The proportional allocation rule is obviously symmetric and strongly
budget-balanced.  It is sybil-proof for risk-neutral miners and hence
also non-risk-neutral miners (see the discussion in
Section~\ref{ss:axioms2}). 

To complete the proof, assume for the purposes of contradiction that
the proportional allocation rule is not collusion-proof under
arbitrary transfers against the set $\calU$ of convex utility
functions~$U$ that satisfy $U(1)>U(0)=0$.
Then, there is a configuration
$\bfh\in \bN^n$, utility functions $U_1,U_2,\ldots,U_n \in \calU$
for the participating miners,
a subset $T\subseteq [n]$ of the miners, and a reward-sharing scheme
$\{\bfr_i \}_{i\in T}$ such that
\[
\frac{\sum_{j \in T} h_j}{\sum_{j \in [n]} h_j}  \cdot
\expect{U_i(\br_i)} \ge \frac{h_i}{\sum_{j \in [n]} h_j} \cdot U_i(1)\quad\text{for every $i\in T$},
\]
with the inequality strict for at least one miner $i\in T$.

Summing these inequalities over~$i \in T$ and using that $U_i(1)>0$
for every miner~$i$, we then have
\[
\frac{\sum_{j \in T} h_j}{\sum_{j \in [n]} h_j} \cdot \sum_{i\in T}
\expect{\frac{U_i(\bfr_i)}{U_i(1)}} >\sum_{i\in T} \frac{h_i}{\sum_{j
    \in [n]} h_j}
\]
and, thus,
\[
\expect{\sum_{i\in T}\frac{U_i(\bfr_i)}{U_i(1)}}>1.
\]
However, since each $U_i$ is convex with $U_i(1)>U_i(0)=0$, we have
  $U_i(\br_i)\le \br_i\cdot U_i(1)$ (with probability~1) and thus, since $\sum_{i \in T}
  \bfr_i \le 1$ (again with probability~1),
\[
\expect{\sum_{i\in T}\frac{U_i(\bfr_i)}{U_i(1)}} \le \expect{\sum_{i\in T}\bfr_i}\le 1.
\]
This completes the contradiction and the proof.
\end{proof}

\section{Acknowledgments}

All three authors were supported in part by the Columbia-IBM Center
for Blockchain and Data Transparency.  The first author was also supported in part by 
  NSF IIS-1838154 and NSF CCF-1703925. 
The second author was supported in part by NSF CCF-1763970.
The third author was 
supported in part by NSF Award CCF-1813188 and ARO grant W911NF1910294.

\begin{flushleft}
\bibliography{axiomatic}
\bibliographystyle{acm}
\end{flushleft}

\end{document}